\documentclass{llncs}
\usepackage{algorithm}
\usepackage[noend]{algorithmic}
\usepackage{amssymb,latexsym}
\usepackage{amsmath}
\usepackage{pst-tree,pspicture}
\usepackage{url}
\usepackage{multirow,subfigure}
\usepackage[section]{placeins}
\newcommand{\st}{\delta}

\newcommand{\LL}{\textsl{L}} %\LL
\newcommand{\RR}{\textsl{R}}
\newcommand{\ass}{:=}
\begin{document}

\title{From Total Assignment Enumeration\\ to Modern SAT Solver%
\thanks{This research was supported
by the Israel Science Foundation (grant no.\ 250/05).
The work of Alexander Nadel was carried out in partial fulfillment of the
requirements for a Ph.D.}} % restore in final vs.

\author{Nachum Dershowitz\inst{1} \and Alexander
Nadel\inst{1,2}}

\institute{School of Computer Science, Tel Aviv University, Ramat
Aviv,
Israel\\
\email{\{nachumd,ale1\}@tau.ac.il} \and
Design Technology Solutions Group, Intel Corporation, Haifa, Israel}
\maketitle
\begin{center}
\bf Draft 2009
\end{center}
\pagestyle{plain} % remove

\begin{abstract}
A new framework for presenting and analyzing the functionality of a modern DLL-based SAT solver is proposed. Our approach exploits the inherent relation between backtracking and resolution. We show how to derive the algorithm of a modern SAT solver from DLL step-by-step. We analyze the inference power of Boolean Constraint Propagation, Non-Chronological Backtracking and 1UIP-based Conflict-Directed Backjumping. Our work can serve as an introduction to a modern SAT solver functionality and as a basis for future work on the inference power of a modern SAT solver and on practical SAT solver design.
\end{abstract}

% -- -- -- -- -- -- -- -- -- -- -- -- -- -- -- -- -- -- -- -- -- -- -- -- -
\section{Introduction}

Propositional satisfiability (SAT) is the problem of determining for a formula in propositional calculus, whether there exists a satisfying assignment for its variables. This problem belongs to a large family of NP-complete problems. SAT has numerous applications, e.g.\ in formal verification~\cite{DBLP:journals/sttt/PrasadBG05}.  Modern complete SAT solvers, based on the original backtrack search algorithm DLL~\cite{dll}, are able to efficiently solve SAT instances arising in real-world applications. DLL was studied and enhanced over the years (see~\cite{FreemanPhD} for an overview), however a major breakthrough was made by the authors of the GRASP SAT solver~\cite{GRASP1},  making it practically efficient.  GRASP introduced a number of innovations in backtracking, united under the title, ``conflict analysis''. These algorithms were further refined in the Chaff solver~\cite{Chaff2001}.

Chaff's conflict analysis, inherited by the most modern SAT solvers (e.g.\ Minisat~\cite{DBLP:conf/sat/EenS03}), includes the following enhancements to DLL: (1) Boolean Constraint Propagation (BCP)~\cite{DP}; (2) Non-Chronological Backtracking (NCB)~\cite{GRASP1}; (3) 1UIP-based Conflict-Directed Backjumping (CDB)~\cite{Chaff2001}; and (4) 1UIP-based Conflict Clause Recording (CCR)~\cite{Chaff2001}.\footnote{Citations are to the first application of these algorithms to SAT. See~\cite{FreemanPhD,SilvaPhD} for an overview of earlier work.} In the existing literature on practical SAT solver design, including \cite{GRASP1,Chaff2001}, the above-mentioned algorithms are considered to be interdependent; they are described and examined together by the means of implication graph analysis~\cite{GRASP1}. 

We show how to add these four enhancements to basic backtracking individually and independently to derive a full-fledged modern solver. Our work can be used as a guide for implementing a modern SAT solver by carrying out a well-defined sequence of steps, summarized in the conclusion, which also includes references to papers on data structures and heuristics for SAT, not discussed here. 

Proof (inference) systems can be compared in terms of the sizes of the shortest proofs (refutations) they sanction~\cite{DBLP:journals/jsyml/CookR79}. We say that $Q$ is \emph{at least as strong as} $P$ ($Q \lesssim P$) if every unsatisfiable CNF formula has a refutation in $Q$ that is no longer than the minimal refutation in $P$. (This is a quasi-ordering.) We say that $P$ and $Q$ are \emph{equally strong} ($P \sim Q$) if minimal refutations of every formula are of the same size. A proof system $Q$ \emph{p-simulates} $P$ if every formula has a refutation in $Q$ that is at most polynomially longer than $P$'s.

\emph{General resolution} is one of the most popular and simplest automatable proof systems. \emph{Tree-like resolution (TLR)} is a restricted version, wherein a proof  takes the form of a tree, rather than a directed acyclic graph (dag). The \emph{size} of a resolution refutation is the number of resolvent clauses generated. DLL-based solvers can also be seen as proof systems, where the size of a proof is the number of decisions made. 

The inference power of DLL with Conflict Clause Recording has been analyzed in a number of recent works~\cite{beame03understanding,DBLP:conf/lpar/Gelder05a,ClRes}. In particular, in~\cite{beame03understanding} it is shown that DLL with CCR and unlimited restarts p-simulates general resolution, where ``restarts'' \cite{restarts} is the technique that allows for restarting the search at any decision point, keeping conflict clauses. However, inference power results depend strongly on the underlying formalization of DLL and CCR. In particular, it has been observed in~\cite{ClRes} that the formalization of CCR, used in~\cite{beame03understanding}, was too general. It allowed the algorithm to continue the search, even if one of the clauses is falsified by the current assignment. Reference~\cite{ClRes} used another model, under which the solver is forced to use BCP and backtracking once a falsified clause is identified, and proved that DLL with CCR can ``effectively'' p-simulate general resolution in a sense made precise in~\cite{ClRes}. The problem of whether or not DLL with CCR can p-simulate general resolution remains open. 

Our framework should be helpful for future work on analyzing the power of DLL with Conflict Clause Recording, but we concentrate on analyzing the power of other algorithms, implemented in modern SAT solvers. To the best of our knowledge, the inference power of BCP, 1UIP-based CDB and NCB has never been examined in literature. (This is surprising, since these algorithms are widely used in modern SAT solvers.) We demonstrate that DLL with 1UIP-based CDB, DLL with NCB, plain DLL and TLR are equally strong. We show that although DLL with BCP p-simulates DLL, there is a formula whose shortest refutation in DLL with BCP is linearly longer than in DLL. We also show that DLL is at least as strong as DLL with NCB, 1UIP-based CDB and BCP, intuitively meaning that a SAT solver without Conflict Clause Recording is not stronger than DLL or TLR. Our results follow from simple analysis of the impact of each algorithm on the resolution refutation construction. 

A fundamental enhancement to the DLL algorithm that should be added before others is \emph{parent clause} maintenance.\footnote{This concept can be traced back to the ``assertion clause'' of~\cite{GRASP1}.} A modern SAT solver associates every flipped literal with a parent clause -- a clause, composed of the flipped literal and a disjunction of a subset of previously assigned literals, negated. Intuitively, the parent clause is a sufficient reason for the flip. It is derived by resolution upon backtracking. A fundamental notion, which we will base our analysis on, is \emph{parent resolution} of a flipped literal, that is, the resolution derivation of the parent clause. 

In previous work~\cite{DBLP:conf/sat/DershowitzHN07}, we proposed comparing and enhancing learning schemes for a modern SAT solver by understanding it as a decision-tree construction engine. 
The current work is based on the well-studied and more general concept of resolution-refutation. In~\cite{DBLP:conf/sat/DershowitzHN07}, we introduced the notion of decision tree pruning, where ``backward pruning'' reduces the size of the newly generated left decision subtree and ``forward'' pruning is a measure for the impact of the Conflict Clause Recording scheme on the subsequent search. The empirical advantage of the 1UIP scheme over other schemes was justified by showing that it contributes to backward and forward pruning more than other schemes do -- both analytically and empirically. This result can easily be understood in the new framework, where the ``left decision subtree'' of~\cite{DBLP:conf/sat/DershowitzHN07} corresponds to the parent resolution of a flipped literal, and the notion of reducing the number of left decision subtree nodes corresponds to reducing the size of a newly generated parent resolution. Here, we reflect on the contribution of the described algorithms to backward search pruning, relating our analysis to the results of~\cite{DBLP:conf/sat/DershowitzHN07}. %  carried out by the algorihtm during backtracking, when nodes are. Formalization and classification of search pruning allowed us providing an analytical explanation of the empirical advantage of the 1UIP learning scheme over other schemes. The search pruning analysis, made in ~\cite{DBLP:conf/sat/DershowitzHN07}, is applicable to the framework of this paper. The underlying framework of this paper is more general than the decision tree-based framework of~\cite{DBLP:conf/sat/DershowitzHN07}, since it relates not only to search pruning, but to various aspects of SAT solver design and proof complexity analysis. 

%We formalize the notion of search pruning, relating it to the potential size of the final resolution refutation of a given unsatisfiable formula, and analyze the impact of various algorithms on search pruning, putting the results of our previous work~\cite{DBLP:conf/sat/DershowitzHN07}, mostly dedicated to search pruning, into a wider context. %We justify the empirical advantage of the 1UIP learning scheme~\cite{Chaff2001} over other schemes, such as UIP-2~\cite{DBLP:conf/sat/DershowitzHN07} and \AllUIP{}~\cite{cdl}. Our previous work~\cite{DBLP:conf/sat/DershowitzHN07} already compared these schemes and proposed an enhancement to the 1UIP scheme, called local conflict clause recording. We use the results of~\cite{DBLP:conf/sat/DershowitzHN07} in our analysis and put them in a wider context. The underlying framework of this paper is more general than the decision tree-based framework of~\cite{DBLP:conf/sat/DershowitzHN07}, since it relates not only to search pruning, but to various aspects of SAT solver design and proof complexity analysis. 

Nieuwenhuis et al.~\cite{NieuwenhuisetalJACM2006} provide a formalization of modern complete SAT algorithms, allowing one to formally reason about their basic properties, such as completeness and termination, in a simple way. Their formalism allows one to easily extend DLL to serve as a basis for algorithms for Satisfiability Modulo Theories (SMT). Our framework is different in that it is meant to be used for practical SAT solver research and proof complexity considerations

Section~\ref{sec:def} provides basic definitions. Sections~\ref{sec:alltlr}--\ref{sec:mssa} show how to construct a modern SAT solver starting from Total Assignment Enumeration. We give some results on inference power and analyze the contributions of various algorithms to backward search pruning. This is followed by conclusions.

\section{Definitions}
\label{sec:def}

We denote \emph{(propositional) variables} by lowercase Latin letters. A \emph{literal} is a variable $v$ or its
negation $\lnot v$. The Boolean values
are denoted 1 and 0.
%For literal $\ell$, $Var(\ell)$ is its variable. 
For variable $v$ and Boolean value $\sigma$, $v^\sigma$ is
the corresponding literal; that is, $v^1=v$ and $v^0=\lnot v$. A 
\emph{Conjunctive Normal Form (CNF) formula} is a set (or conjunction) of
\emph{clauses} $\left\{ C_1, \ldots, C_m \right\}$, each clause being a disjunction
(or multiset) of literals. We assume that the
\emph{input formula} does not contain the empty clause $\square$.

A clause $C$ is a \emph{resolvent} of clauses $D_1$ and $D_2$ on \emph{pivot variable}
$v \in D_1$, denoted  $C = D_1 \otimes^v D_2$, if $\lnot v \in D_2$, and
$C = D_1 \cup D_2 \setminus \left\{v,\lnot v\right\}$. 
Resolvent $C$ is \emph{non-trivial} if 
$D_1$ and $D_2$ are
\emph{non-redundantly resolvable}, in the sense that there is a \emph{pivot variable} $v=pivot(D_1,D_2)$, such that
the resolvent of $D_1$ and $D_2$ on $v$ is not a tautology.

A \emph{general resolution refutation} of a given formula $\alpha = \left\{ C_1, \ldots, C_m \right\}$ is a dag $G_{\alpha}=(\alpha\cup\left\{C_{m+1}, \ldots, C_{m+n}\right\},E)$, whose nodes are (associated with) clauses $C_i$ and whose edges $E$ represent resolution relations between clauses. Nodes  corresponding to initial clauses are the sources of the graph. Each non-source node $C_i$ is associated with a \emph{pivot variable} $p_i$. Each edge $(i,j)\in E$,  from node $C_i$ to $C_j$, has an associated Boolean value $\tau(i,j)$ and a status $\st(i,j)$, which can either be $\LL$ or $\RR$, standing for left and right, respectively.\footnote{Status values make sense in the context of tree-like resolution.}
Each non-source node $C_i$ has two incoming edges $(j,i)$ and $(k,i)$, associated with opposite Boolean values and opposite statuses. 
Nodes at the other side of a left or right incoming edge are called the \emph{left} and \emph{right child} of $C_i$, respectively. 
%Left and right children are \emph{brother} of each other. 
Clauses $C_j$ and $C_k$ are non-redundantly resolvable on $p_i$ and $\tau(j,i) = \rho$ iff $p_i^{\lnot \rho} \in C_j$. A resolution refutation is \emph{complete} if the last clause $C_{m+n}$ is the empty clause $\square$; otherwise, it is \emph{partial}. 
The \emph{size} of a resolution refutation is the number of non-source nodes.

An example of a general resolution refutation, of size 4, appears in Fig.~\ref{fig:grr}.
The corresponding resolvent clause appears at each non-source node. The pivot variables are not shown; instead, the literal $p_i^{\tau(e)}$ labels each edge $e$.

A \emph{tree-like resolution (TLR) refutation} of a formula $\alpha$ is a resolution refutation $G$, such that a non-source clause appears on each path from a source to the target clause only once.
In other words, $G$ without the source nodes forms a tree. 
A \emph{regular resolution refutation} of $\alpha$ is a general resolution refutation, such that pivot variables along each particular path from a source to $\square$ are different. Each node $C_i$ of a valid resolution refutation $G$ is referred to as a \emph{root of a tree-like resolution refutation}, if $C_i$ is a root of a tree in $G$ with only non-source nodes, in which case $G$ is a \emph{resolution derivation} of $C_i$.

\section{SAT Solver Skeleton}
\label{sec:alltlr}

Modern SAT solvers are rooted in Total Assignment Enumeration (TAE) -- a DFS search in the assignment space, checking the satisfiability of each clause only after all variables assigned. The only difference between the DLL algorithm~\cite{dll} and TAE is that DLL checks satisfiability of each clause after every assignment. Both TAE and DLL can be viewed as proof systems for unsatisfiable formulas. We define the proof size for both algorithms as the number of decisions, where flip operations are not considered to be decisions.

\begin{proposition}\label{th:taeb}
TAE does not p-simulate DLL.
\end{proposition}
\begin{proof}
Consider the formula $a \land \lnot a$ over $n$ variables. The size of the shortest DLL proof is 1. The size of any TAE proof is $2^n-1$.
\qed
\end{proof}

\begin{algorithm}[t]
\caption{SAT Solver Skeleton (SSS)}\label{alg:dll}
\begin{algorithmic}[1]
\STATE $\textit{Instance} \ass \left\{C_1,C_2, \ldots, C_m \right\}$
\STATE $d \ass 0$
\STATE $G \ass \textit{InitResolutionRefutation}(Instance)$
\LOOP \label{line:mainloop}
\STATE $\textit{NewParent} \ass none$
\STATE $d \ass d + 1$\label{line:dllpp}
\STATE $\left\langle v_d,\sigma_d \right\rangle \ass \textit{ChooseNewLiteral\/}(v_1,\ldots,v_{d-1})$ \label{line:decide}
\STATE $LRStatus(d) \ass \LL$ \label{line:statusleft}
\IF{$\sigma_{1:d}(Instance)=1$}
\RETURN satisfiable
\ENDIF
\WHILE{$\exists l \in \textit{Instance} \cup \left\{\textit{NewParent}\right\} : \sigma_{1:d}(C_l)=0$} \label{line:cal}
\STATE $Parent(d) \ass l$ \label{line:frds}
\STATE $\sigma_d \ass \lnot \sigma_d$ \label{line:flip}
\STATE $LRStatus(d) \ass \RR$ \label{line:statusright}
\IF{$\exists r \in \textit{Instance}: \sigma_{1:d}(C_r)=0$}\label{line:ifconfafterflip}
\STATE $\textit{NewParent} \ass r$
\WHILE{$d > 0$ and ($LRStatus(d)=\RR$ or $v_d^{\lnot \sigma_d} \notin C_{\textit{NewParent}}$)} \label{line:bal}
\IF{$LRStatus(d)=\RR$ and $v_d^{\lnot \sigma_d} \in C_{\textit{NewParent}}$ }\label{line:litincond}
\STATE $\textit{NewParent} \ass \textit{AddNode}(G,Parent(d),\textit{NewParent},v_d)$\label{line:addtoresref}
\ENDIF
\STATE $d \ass d - 1$\label{line:dminus1}
\IF{$d = 0$}\label{line:dis0}
\RETURN unsatisfiable\label{line:dllreturn}
\ENDIF
\ENDWHILE
\ENDIF
\ENDWHILE
\ENDLOOP
\end{algorithmic}
\end{algorithm}

\begin{figure}[t]
%\centering 
\scriptsize
\subfigcapskip=3pt
\renewcommand{\subcapsize}{\scriptsize}
%\hspace{22mm}
\subfigure[{A general resolution refutation $G$ of $\alpha$ or of $\alpha'$}]{ \label{fig:grr}
\psmatrix[colsep=7pt,rowsep=9pt]
&&&$\square$& \\
&$a$&&&$\lnot a$\\
$a \lor b$&&$\lnot b$&&$\lnot a \lor b$\\
&$\lnot b \lor c$&&$\lnot b \lor \lnot c$& \\
\endpsmatrix
\psset{shortput=nab,labelsep=3pt,arrows=<-} 
\ncline{1,4}{2,2}\Bput{$\lnot a$} 
\ncline{1,4}{2,5}\Aput{$a$} 
\ncline{2,2}{3,1}\Bput{$\lnot b$} 
\ncline{2,2}{3,3}\Aput{$b$} 
\ncline{3,3}{4,2}\Bput{$\lnot c$}
\ncline{3,3}{4,4}\Aput{$c$}
\ncline{2,5}{3,3}\Bput{$b$}
%\ncline[linestyle=dotted]{2,5}{3,3}\Bput{$b$}
\ncline{2,5}{3,5}\Aput{$\lnot b$}}
\subfigure[{Snapshot of invocation  of Algorithm~\ref{alg:dll} on $\alpha$}]{ \label{fig:sdll}
\psmatrix[colsep=7pt,rowsep=9pt]
&&& \\
&$$&&\\
$a \lor b$&&$$&\\
&$\lnot b \lor c$&&$\lnot b \lor \lnot c$ \\
\endpsmatrix
\psset{shortput=nab,labelsep=3pt,arrows=-} 
\ncline{1,4}{2,2}\Bput{$\lnot a$} 
\ncline{2,2}{3,1}\Bput{$\lnot b$} 
\ncline{2,2}{3,3}\Aput{$b$} 
\ncline{3,3}{4,2}\Bput{$\lnot c$}
\ncline{3,3}{4,4}\Aput{$c$}
}
\subfigure[{Backtracking and flipping, given Fig.~\ref{fig:sdll}}]{ \label{fig:abf}
\psmatrix[colsep=7pt,rowsep=9pt]
&&&$$& \\
&$a$&&&$$\\
$a \lor b$&&$\lnot b$&&\\
&$\lnot b \lor c$&&$\lnot b \lor \lnot c$& \\
\endpsmatrix
\psset{shortput=nab,labelsep=3pt,arrows=-} 
\ncline{1,4}{2,2}\Bput{$\lnot a$} 
\ncline{1,4}{2,5}\Aput{$a$} 
\psset{shortput=nab,labelsep=3pt,arrows=<-} 
\ncline{2,2}{3,1}\Bput{$\lnot b$} 
\ncline{2,2}{3,3}\Aput{$b$} 
\ncline{3,3}{4,2}\Bput{$\lnot c$}
\ncline{3,3}{4,4}\Aput{$c$}
}
\subfigure[{Snapshot of invocation of Algorithm~\ref{alg:dll} on $\alpha'$}]{ \label{fig:s2}
\psmatrix[colsep=7pt,rowsep=9pt]
&&$$&&\\
&$$&&&\\
$a \lor \lnot b$&&$$&&\\
&$\lnot b \lor c$&&$\lnot b \lor \lnot c$&\\
\endpsmatrix
\psset{shortput=nab,labelsep=3pt,arrows=-} 
\ncline{1,3}{2,2}\Bput{$b$}
\ncline{2,2}{3,1}\Bput{$\lnot a$}
\ncline{2,2}{3,3}\Aput{$a$}
\ncline{3,3}{4,2}\Bput{$\lnot c$}
\ncline{3,3}{4,4}\Aput{$c$}
}
\subfigure[{1UIP-based CDB and flipping, given Fig.~\ref{fig:sdll}; backtracking and flipping, given Fig.~\ref{fig:s2}}]{ \label{fig:after}
\psmatrix[colsep=7pt,rowsep=9pt]
&&&$$&&&&&&&&& \\
&$\lnot b$&&&$$&&&&&&&&\\
$\lnot b \lor c$&&$\lnot b \lor \lnot c$&&&&&&&&&&\\
%&&&& \\
\endpsmatrix
\psset{shortput=nab,labelsep=3pt,arrows=-} 
\ncline{1,4}{2,2}\Bput{$b$} 
\ncline{1,4}{2,5}\Aput{$\lnot b$} 
\psset{shortput=nab,labelsep=3pt,arrows=<-} 
\ncline{2,2}{3,1}\Bput{$\lnot c$} 
\ncline{2,2}{3,3}\Aput{$c$} 
}
\subfigure[{NCB effect}]{ \label{fig:ncb}
\psmatrix[colsep=7pt,rowsep=9pt]
&&&$$&&&$$\\
&&$$&&&$$&\\
&$$&&&$a \lor b$&&$$\\
$a \lor b$&&&&&&\\
\endpsmatrix
\psset{shortput=nab,labelsep=3pt,arrows=-} 
\ncline{1,4}{2,3}\Bput{$\lnot a$}
\ncline{2,3}{3,2}\Bput{$c$}
\ncline{3,2}{4,1}\Bput{$\lnot b$}
\ncline{1,7}{2,6}\Bput{$\lnot a$}
\ncline{2,6}{3,5}\Bput{$\lnot b$}
\ncline{2,6}{3,7}\Aput{$b$}
}
\caption[Examples of search trees, resolution refutations and the impact of various algorithms.]{Examples for $\alpha = (a \lor b) \land (\lnot b \lor c) \land (\lnot b \lor \lnot c) \land (\lnot a \lor b)$ and $\alpha' = \alpha \land (a \lor \lnot b)$} \label{fig:ex}
\end{figure}

Algorithm~\ref{alg:dll}, which we refer to as the \emph{SAT Solver Skeleton (SSS)}, is an implementation of DLL, enhanced by parent clause and parent resolution maintenance.

First, we depict the general flow of Algorithm~\ref{alg:dll}. The algorithm comprises three loops: the main loop (starting at line~\ref{line:mainloop}), the conflict analysis loop (line~\ref{line:cal}) and the backtracking loop (line~\ref{line:bal}). Each iteration of the main loop increases the decision level $d$ and assigns an unassigned decision variable $v_d$ to some value $\sigma_d$. If the formula is satisfied, the algorithm returns. Otherwise, if none of the clauses is falsified by the current assignment $\sigma_{1:d}$, the main loop continues. (We denote by $\sigma_{1:d}$ the partial assignment induced by assignments to decision variables corresponding to decision levels $1 \ldots d$.) If one of the clauses $C_l$ is falsified by $\sigma_{1:d}$, the algorithm enters the conflict analysis loop. In this case, we say that a \emph{conflict} takes place in a \emph{blocking clause} $C_l$. The conflict analysis loop continues working until a new decision is required or the formula is proved to be unsatisfiable. As a first step, it flips the value of $v_d$. If no conflict follows, a new decision is required and the algorithm exits the conflict analysis loop and returns to the main loop. If a conflict follows, the algorithm enters the backtracking loop. The backtracking loop is responsible for backtracking to the lowest possible decision level $d$, whose decision variable can be flipped. The backtracking loop may also prove that no such decision level exists, in which case the formula is unsatisfiable. 

A decision level is \emph{left} before its decision variable has been flipped,
and \emph{right} after. The status of each decision level $d$ is maintained in $LRStatus(d)$ in Algorithm~\ref{alg:dll}. The algorithm maintains a \emph{parent clause} $Parent(d)$ for each right decision level $d$, which must be a logical consequence of the initial formula, and it must consist of the literal $v_d^{\sigma_d}$ and a subset of literals, falsified by $\sigma_{1:d-1}$. Intuitively, the parent clause explains why $v_d$ was flipped. It can be seen as an implication $\alpha \Rightarrow v_d$, where $\alpha$ is a conjunction of a subset of variables assigned before $v_d$. The parent clause is derived by tree-like resolution during backtracking. The derivation of the parent clause is the \emph{parent resolution} of $v_d$. The following two invariants must hold throughout execution of SSS:

\begin{enumerate}
    \item \textbf{Flip-consistency:} For each right decision level $d$, $Parent(d)$ is a valid parent clause. \label{inv:1}
    \item \textbf{Resolution-consistency:} $G$ is a valid TLR refutation; and for each right decision level $d$, the node $Parent(d)$ is a root of a valid tree-like resolution refutation.
\end{enumerate}

Now we describe the parent resolution and parent clause creation process, demonstrating that the two invariants hold.\footnote{For a formal proof, see~\cite{NadelPhD}.} The parent clause is set for each flip at the beginning of the conflict analysis loop (line~\ref{line:frds}). Suppose that we are in the first iteration of the conflict analysis loop. The clause $C_l \in \textit{Instance}$ is falsified by $\sigma_{1:d}$ before the flip. (Note that $\textit{NewParent}$, whose usage will be described shortly, is always $none$ at the first iteration of the conflict analysis loop.) It is easy to check that $C_l$ is a valid parent clause and it is the root of a trivial tree-like resolution refutation. %Indeed, it must contain $v_d$, since otherwise it was falsified during previous iterations of the main loop. Also, its other literals must have been falsified by $\sigma_{1:d-1}$. The clause $C_l$ is a logical consequence of the initial formula, since it belongs to the initial formula; and it is a root of a trivial tree-like resolution refutation for the same reason.

Now we analyze the case when the parent clause is created by the backtracking loop. Suppose that there is a conflict after the flip, made during the first iteration of the conflict analysis loop. The backtracking loop maintains a \emph{backtracking clause} $\textit{NewParent}$. Each iteration of the backtracking loop maintains the \textit{backtracking invariant}: $\sigma_{1:d-1}(\textit{NewParent} \setminus \left\{v_d^{\lnot \sigma_d}\right\})=0$ and $\textit{NewParent}$ is a root of a tree-like resolution refutation. Observe that if the backtracking invariant holds, then the flip-consistency and the resolution-consistency invariants will hold after the backtracking loop finishes. 

Before the first iteration, the backtracking clause is initialized to the newly discovered blocking clause. The parent clause of $d$ and the blocking clause, encountered after flipping $v_d$, are non-redundantly resolvable. Thus, a new valid node is added to $G$ by the algorithm.\footnote{The function $\textit{AddNode}$ takes a valid resolution refutation $G$, a left son, a right son and a pivot variable, and creates a new node in $G$, returning its index.} The resulting clause $Parent(d) \otimes^{v_d} \textit{NewParent}$ becomes the new backtracking clause. When the algorithm visits a decision level $d$ on subsequent iterations of the backtracking loop, one of the following cases happen: 
\begin{enumerate}
  \item The decision level $d$ is 0. In this case, the formula is unsatisfiable, and the backtracking clause must be $\square$ by the backtracking invariant.
	\item The decision level $d$ is left and the negation of its decision literal belongs to the backtracking clause ($LRStatus(d)=\LL$ and $v_d^{\lnot \sigma_d} \in C_{\textit{NewParent}}$). The backtracking loop terminates as it has found a variable to flip and has built its parent clause and resolution.
	\item The decision level $d$ is right and $v_d^{\lnot \sigma_d} \in C_{\textit{NewParent}}$. The backtracking loop resolves the parent clause of $d$ with the backtracking clause to receive a new backtracking clause. One can easily verify that: (1) $C_{Parent(d)}$ and $C_{\textit{NewParent}}$ are non-redundantly resolvable with pivot variable $v_d$; (2) the new backtracking clause must be falsified by $\sigma_{1:d-1}$, and it must be a root of a TLR refutation.%This is a valid non-redundant resolution operation with pivot variable $v_s$, since: (1) It must hold that $v_d^{\sigma_d} \in Parent(d)$ by the flip-consistency invariant and $v_d^{\lnot \sigma_d} \in C_{\textit{NewParent}}$ holds by our condition; (2) The other literals must have 
	\item The decision level $d$ is left and $v_d^{\lnot \sigma_d} \notin C_{\textit{NewParent}}$. In this case, the backtracking loop of SSS does not flip $v_d$ and continues backtracking. Indeed, the backtracking clause must be falsified by $\sigma_{1:d-1}$; thus there is no satisfying assignment under $v_1=\sigma_1, \ldots, v_{d-1}=\sigma_{d-1}, v_d=\lnot \sigma_d$. The behavior of our algorithm in this case shows the difference between SSS and plain DLL, which flips every left decision variable.
	\item The decision level $d$ is right and $v_d^{\lnot \sigma_d} \notin C_{\textit{NewParent}}$. The algorithm backtracks to the next decision level without carrying out the resolution operation. We say that \emph{resolution backward pruning} takes place in this case. We relate search pruning to the ability of the algorithm to reduce the number of nodes in the final resolution refutation of the formula. In our case, the parent resolution of $v_d$ is not included in the derivation of the new backtracking clause; thus it will not be included in the derivation of the newly flipped variable, which in turn means that it will not be included in the final resolution refutation of the formula. Resolution backward pruning corresponds to one of the three cases of backward tree pruning of~\cite{DBLP:conf/sat/DershowitzHN07} (``skipping of inactive lf-variables, not connected to the conflicting clause vertices''). We will encounter the other two types of backward pruning when discussing NCB and 1UIP-based CDB.
\end{enumerate}

Consider the snapshot of an SSS invocation after the second conflict in Fig.~\ref{fig:s2}. The current decision level is 3. In the first iteration of the backtracking loop, a new clause $\lnot b = \lnot b \lor c \otimes^c \lnot b \lor \lnot c$ is created and the decision level becomes 2. The right decision variable $a$ does not appear in the newly created clause; hence backward pruning takes place. Backtracking continues and no new clause is created inside the backtracking loop during this iteration. The backtracking stops at decision level 1, as the clause $\lnot b$ will be a valid parent clause after flipping the variable $b$. The situation that results after the flip appears in Fig.~\ref{fig:after}. The bottom-left part of the figure, which includes nodes with clauses and arrowed edges, represents the parent resolution of $b$, created by the backtracking loop. Note that the parent resolution of $\lnot a$, which consists of the single clause $a \lor \lnot b$, does not appear in the new parent resolution. Another example of backtracking and flipping is the transformation from Fig.~\ref{fig:sdll} to Fig.~\ref{fig:abf}.

\section{A Tree-Like SAT Solver}
\label{sec:b1uip}

Next, we show how to augment Algorithm~\ref{alg:dll} with Boolean Constraint Propagation, 1UIP-based Conflict-Driven Learning and Non-Chronological Backtracking -- separately and independently. We analyze the inference power of each algorithm. We begin by showing that Algorithm~\ref{alg:dll}, DLL and TLR are equally strong.

%There are three differences between Algorithm~\ref{alg:dll} and tree-like Chaff: Boolean Constraint Propagation (BCP); Non-Chronological Backtracking (NCB) and 1UIP-based Conflict-Directed Backjumping (CDB). We show how to add these enhancements step-by-step to Algorithm~\ref{alg:dll}, separately and independently. We use backtracking and resolution-related terminology instead of implication graph-related terminology. We show that DLL with NCB and 1UIP-based CDB is equally strong to TLR and provide a family of formulas, linearly separating DLL from DLL enhanced by BCP.

\subsection{The Power of SAT Solver Skeleton}
\label{sec:ppsss}

The only difference between SSS and DLL in terms of search space exploration is the fact that DLL flips every left variable, whereas SSS may skip flipping some variables. We show that parent clause and resolution maintenance do not change the inference power of DLL. It remains the same as TLR.\footnote{The fact that DLL is identical to TLR is well-known.} This observation means that parent clause and resolution maintenance is a heuristic, enabling the finding of shorter proofs by compressing a proof on-the-fly.

\begin{proposition}\label{th:dlltrr}
TLR $\sim$ SSS $\sim$ DLL.
\end{proposition}
\begin{proof}
We consider DLL to be a simplified version of Algorithm~\ref{alg:dll}, which flips every left decision variable and does not maintain parent clauses and resolution refutations. 
We will prove in turn that: TLR $\gtrsim$ SSS; SSS $\gtrsim$ DLL; DLL $\gtrsim$ TLR. 

Consider the shortest tree-like resolution refutation $H$ of size $k$ of any unsatisfiable formula. We show that there exists an SSS invocation, whose size is at most $k$. We let the SSS algorithm explore the reversed dag $\overline{H}$ in a depth-first (DFS) manner, starting with $\square$ assigning literals, associated with the edges of $H$. We denote the currently visited node of $H$ by $C_h$. We enforce the SSS algorithm to always choose clauses appearing at the leaves of $H$ as blocking clauses in case of ambiguity. It is sufficient to show that the following invariants always hold: (1) A conflict is encountered by SSS iff a leaf of $H$ is reached; (2) SSS will flip exactly the variable that DFS backtracks to. First, observe that the second invariant must hold unless the first one is violated. Indeed, the backtracking clause must be the visited clause of $H$, given that $H$ is a valid TLR refutation. Second, note that if a leaf $C_h$ of $H$ is reached, a conflict must be found by SSS in $C_h$, since all the literals of $C_h$ are assigned 0 by construction. Finally, we show that if a conflict is found by SSS in a clause $D$, then $C_h$ must be a leaf. Suppose to the contrary, $C_h$ is not a leaf. Denote the only path from $C_h$ to $\square$ by $M = \left\{M_1, \ldots, M_k \right\}$, where 
$M_1=C_h$ and $M_k=\square$. Then, the following operations would transform $H$ to a TLR refutation shorter than $H$: (1) Replace $C_h$ by $D$ in $H$ and delete the derivation of $C_h$ from $H$. This operation decreases the size of $H$, however $H$ is no more a TLR, unless $C_h=D$; hence we need to ``fix'' it. (2) For every literal $l \in D \setminus C_h$, augment every clause of $M$ with $l$ starting with $M_1=C_h$, until a node with pivot variable $l$ is reached. The last condition must hold, since otherwise $l$ would not be assigned; (3) Remove the literals of $C_h \setminus D$ from the clauses of $M$. This step might leave unnecessary nodes in $M$ -- nodes, one of whose sons does not contain the pivot variable. (4) Remove the unnecessary nodes from $M$ together with the resolution derivation of the son that does not appear in $M$. 

Now consider the shortest SSS invocation of size $k$. There exists a DLL invocation of at most the same size, since the shortest invocation of SSS must flip every left variable. Indeed, if a left variable $b$ was skipped, then not making this decision would result in an SSS invocation of size $k-1$.

Finally, consider the shortest DLL invocation. We show that there exists a TLR refutation of at most the same size. Consider an invocation of SSS taking the same decisions. Such SSS invocation is valid, and it refutes the given formula, since it a left decision level cannot be  skipped, otherwise the DLL invocation would not be the shortest one. SSS outputs a tree-like resolution of at most the same size, since any node of the TLR refutation corresponds to a backtracking step, and backtracking steps correspond to decisions in a one-to-one manner.
\qed
\end{proof}

\subsection{Boolean Constraint Propagation (BCP)}
\label{sec:BCP}

A clause $C$ is \emph{a unit clause} at decision level $d$ if $C$ evaluates to a lone literal $v^{\rho}$ under $\sigma_{1 : d}$.

\emph{Boolean Constraint Propagation (BCP)} is the following process, carried out by the solver at each decision point: If there is a unit clause $C$ at level $d$, pick the opposite literal $v^{\lnot\rho}$ as the next decision. Observe that the algorithm would then encounter a conflict and would flip the value of $v$ automatically in the conflict analysis loop. It is accepted in the literature to refer to $C$ after this operation as a parent clause of an \emph{implied literal\/} $v^{\rho}$. In our formulation, implied literals are treated as regular flipped decision variables, in contrast to the separation between decision and implied variables. This approach allows us to avoid implication graph terminology. 

It is widely accepted that BCP helps accelerate modern SAT solvers, though it typically consumes 80--90\% 
of a solver's run-time~\cite{Chaff2001}. The added value of BCP is that it allows the solver to quickly propagate information and find conflicts. However, this claim is accurate only when unit clauses, identified by BCP, are relevant for the resolution process. We show below that BCP can decrease the inference power of DLL by a linear factor; therefore, at least in some cases, BCP may slow down the solver by making unnecessary propagations. Nevertheless, as we will see, the damage is never exponential. 
%cannot be greater than linear, since at most $z$ unnecessary decisions may be made along each path from the root to the source of the tree, if BCP is invoked. %Thus, it may be found that in practice it is more efficient to use BCP selectively. Separating between BCP and other CDL algorithms, as we have done in this paper, allows one to  design a SAT solver with a modern CDL engine, without any BCP or with only a selective BCP.

To implement BCP, do the following:
\begin{algorithmic}
\STATE \textbf{BCP} (invoked instead of line~\ref{line:decide} of Algorithm~\ref{alg:dll}):\label{line:gbcp1}
\IF{$\exists C_i \in \textit{Instance}: C_i$  is a unit clause $v^{\rho}$ at $d$} 
\STATE $\left\langle v_d,\sigma_d \right\rangle \ass \left\langle v, \lnot \rho \right\rangle$ \label{line:gbcp2}
\ELSE  
\STATE $\left\langle v_d,\sigma_d \right\rangle \ass \textit{ChooseNewLiteral\/}(v_1,\ldots,v_{d-1})$ \label{line:gbcp3}
\ENDIF
\end{algorithmic}

%As BCP is only a special decision strategy for Algorithm~\ref{alg:dll}, the algorithm still terminates and generates a valid TLR refutation. We show that BCP can hurt the inference power of Algorithm~\ref{alg:dll} by a factor linear in the number of variables $z$. The damage cannot be greater than linear, since at most $z$ unnecessary decisions may be made along each path from the root to the source of the tree, if BCP is invoked.

\begin{proposition}\label{th:bcp}
There is a formula whose shortest refutation in DLL with BCP is linearly longer than in DLL.
\end{proposition}

\begin{proof}
Consider a formula consisting of (1) eight clauses, each of size 3, corresponding to all possible disjunctions between literals of variables: $a, b, c$, and (2) the following set of $k$ clauses for each literal $p \in D = \left\{a, b, c, \lnot a, \lnot b, \lnot c \right\}$: $C^p = (p \lor l_1^p) \land (\lnot l_1^p \lor l_2^p) \land (\lnot l_2^p \lor l_3^p) \land \ldots \land (\lnot l_{k-1}^p \lor l_k^p)$. The variables $L^p = \left\{ l_1^p \ldots l_k^p \right\}$ are fresh variables for each of $D$'s literals. 

Clearly, there exists an invocation of DLL refuting the formula with 7 decisions, which ignores clause set (2). 
BCP, however, forces $k$ additional, useless inferences. More specifically, if $p$ is the first literal of $D$ that is assigned, then all the literals of $L^{\lnot p}$ are assigned either before $p$ or as a result of BCP, after $p$'s assignment. 

The complexity of every invocation DLL+BCP on this example is ${\rm\Omega}(3 + 6k)$, compared with constant complexity of DLL. Hence, our formula linearly separates DLL from DLL+BCP.
\qed
\end{proof}

\begin{proposition}\label{th:bcp2}
DLL + BCP p-simulates DLL.
\end{proposition}

\begin{proof}
BCP may add only a linear number of decisions per leaf.
\qed
\end{proof}

\subsection{Non-Chronological Backtracking (NCB)}
\label{sec:NCB}

Non-Chronological Backtracking (NCB) is a backward pruning technique, applied immediately after a new variable for flipping is discovered by the backtracking loop.

Suppose that the algorithm is about to flip a certain left decision variable $v_d$ after finding a corresponding parent clause $C_l$. It may be the case that $C_l$ would still be a parent clause, consisting of $v_d$ and falsified literals, even if one decreased $d$ prior to the flip operation undoing some of the previously made decisions. NCB is the process of backtracking to a highest decision level $g$, so that $C_l$ is still a parent clause. 

After the above-described operation, the NCB implementation of Chaff also increases $d$ up to the closest left decision level. This step is carried out so as not to redo BCP. To implement NCB, do the following:
\begin{algorithmic}
\STATE \textbf{Non-Chronological Backtracking (NCB)} (invoked just before line~\ref{line:frds}): \label{line:gncb1}
\STATE $g \ass$ Minimal $g$, such that $\sigma_{1:g}(C_l \setminus \left\{v_d^{\lnot \sigma_d}\right\})=0$ 
\STATE $g \ass$ First left decision level $\geq g - 1$\COMMENT{An optional step}
\STATE $v_{g+1} \ass v_d; \sigma_{g+1} \ass \sigma_d; d \ass g+1$ \label{line:gncb2}
\end{algorithmic}

The NCB algorithm induces the second type of backward search pruning, which we call \emph{NCB backward search pruning} (``skipping lu-variables'' in~\cite{DBLP:conf/sat/DershowitzHN07}). Recall that we have introduced resolution backward search pruning in Sect.~\ref{sec:alltlr}. If there exist right decision levels between $g$ and $d$, the algorithm does not include their parent resolutions in the parent resolution of the flipped variable; thus these parent resolutions will not be part of the final resolution refutation of the given formula.

Figure~\ref{fig:ncb} shows the effect of NCB. A snapshot of an SSS invocation after the first conflict is depicted on the left-hand side. The algorithm identifies the fact that it can flip the value of variable $b$ at decision level 2, rather than at 3, since the conflict does not depend on the value of $c$. Hence, it unassigns $c$ before the flip. The situation that results appears on the right-hand side of Fig.~\ref{fig:ncb}. Observe that NCB backward pruning does not occur in this example, since the algorithm does not skip right decision levels.

%NCB is constructed in a way that does not violate Theorem~\ref{th:dll}. The inference power of DLL with NCB remains the same: 
%More specifically, Lemma~\ref{lemma:blooplemma} is not impacted by NCB at all; and Lemma~\ref{lemma:clooplemma} still holds, since after applying NCB and flipping, $C_{\fr(d)}$ will still be a parent clause by construction; and termination function still increases.
Now we show that NCB does not change the inference power of SSS.

\begin{proposition}\label{th:dllncb}
SSS + NCB $\sim$ DLL.
\end{proposition}
\begin{proof}
NCB cannot be applied in a shortest SSS invocation, since this would yield the existence of decisions that could be skipped by the shortest invocation of SSS, which is impossible. Thus, an invocation of SSS with NCB taking the same decisions as a shortest SSS invocation is valid. (No actual points for making a Non-Chronological Backtracking exist.) On the other hand, SSS with NCB always generates a TLR refutation of at most the same size.
\qed
\end{proof}

\subsection{1UIP-based Conflict-Directed Backjumping (CDB)}
\label{sec:CDB}

1UIP-based Conflict-Directed Backjumping (CDB) is yet another backward search pruning technique.

A Unique Implication Point (UIP)~\cite{GRASP1} is a well-known concept, whose name is rooted in the implication-based approach to conflict analysis. First, we express this notion in our framework.

A left decision block of a left decision level $h$, $LDB(h)$ is a subset of decision levels that includes $h$ and every right decision level, assigned after $h$, but before the next left decision level (if available).\footnote{Our definition of decision level corresponds to that of GRASP~\cite{GRASP1}. Chaff's~\cite{Chaff2001} decision level is what we call a left decision block.}

Suppose that the SSS is backtracking over a decision level $d$. Let $g$ be the highest left decision level. A right decision variable $v_{d}$ is a \emph{Unique Implication Point (UIP)}, if $v_d$ is the only variable assigned at $LDB(g)$ that appears in the backtracking clause $\textit{NewParent}$. Backtracking may find more than one UIP. UIPs are counted according to their order during the backtracking phase. 

1UIP-based CDB is the following technique: once the first UIP (1UIP) variable $v_d$ is discovered during backtracking, continue as if $v_d$ was a left decision variable, assigned instead of $v_g$, whose parent clause is the current backtracking clause. One way to think about 1UIP-based CDB is as substituting the decision $v_g$ by $v_d$ a-posteriori. Note that a left decision variable can never be a UIP in our notation.

%We found that it is useful to differentiate between 1UIP-based Conflict-Directed Backjumping (CDB), analyzed in this section, and 1UIP-based Conflict Clause Recording. The two techniques are closely related, but can be separated both in theoretical analysis and practical implementation.

%Suppose that the algorithm is backtracking, and has just added a new $\textit{NewParent}$ node to the resolution refutation at right decision level $d$. The idea is to treat $d$ as a left decision level with parent clause $\textit{NewParent}$, substituting the decision $v_g$ by $v_d$; and abandoning the current parent resolution refutation $Parent(d)$. 

%Suppose that $C_{\textit{NewParent}}$ contains only one variable, 1UIP, assigned after the maximal left decision level $g$. Then, the decision variable of $g$ can be substituted by 1UIP with left child $C_{\textit{NewParent}}$.
This is implemented as follows:
\begin{algorithmic}
\STATE \textbf{1UIP-based CDB} (invoked just after line~\ref{line:dminus1}):  \label{line:g1uip1}
\STATE $g \ass$ Highest left decision level
\IF{$v_d^{\lnot \sigma_d} \in C_{\textit{NewParent}}$ and $\sigma_{1:g-1}(C_{\textit{NewParent}} \setminus \left\{v_d^{\lnot \sigma_d}\right\}) = 0$}
\STATE $v_g \ass v_d; \sigma_g \ass \sigma_d; Parent(g) = \textit{NewParent}; d \ass g$ \label{line:g1uip2}
\ENDIF
\end{algorithmic}

See the transformation of Fig.~\ref{fig:sdll} into Fig.~\ref{fig:after} for an example of the effect of 1UIP-based CDB. After the algorithm learns a new resolvent clause $\lnot b$ during backtracking, it discovers that it contains only one variable, $b$, assigned after the highest left decision level 1. So, it substitutes $b$ for $\lnot a$. The parent clause and parent resolution are updated to the backtracking clause and its derivation.

1UIP-based CDB induces the third type of backward search pruning, which we call \emph{UIP backward search pruning} (``skipping of inactive lf-variables, connected to the conflicting clause vertices, but not dominated by the pivot variable'' \cite{DBLP:conf/sat/DershowitzHN07}). Consider a right variable $v_j$ of the last left decision block, such that $j \leq d$, where $d$ is the decision level of the UIP variable. Its parent resolution is not included in the newly derived parent resolution; thus it will not be included in the final resolution refutation. In our example, the parent resolution of $\lnot b$ that consists of a single clause $a \lor b$ is pruned.
%Lemma~\ref{lemma:blooplemma} still holds even if 1UIP-based CDB is applied, since we constructed $C_t$ to be a strict falsification clause, serving as a root of a TLR refutation. Lemma~\ref{lemma:clooplemma} also holds as termination function is guaranteed to increase. Indeed, we flip the decision variable of the last left decision level, thus each flip brings us to a higher decision level. The other conditions hold by construction. Thus, theorem~\ref{th:dll} still holds.
We underscore the fact that we do not consider 1UIP-based Conflict Clause Recording in this section, but only 1UIP-based CDB. We will see that these two concepts are not necessarily related.

The inference power of SSS with 1UIP-based CDB remains the same.

\begin{proposition}\label{th:dll1uip}
SSS + 1UIP-based CDB $\sim$ DLL.
\end{proposition}
\begin{proof}
%We use the fact that DLL $\sim$ SSS $\approx$ TLR (proposition~\ref{}). 

SSS with 1UIP-based CDB is not more powerful than SSS, since it always produces a TLR refutation of at most the same size, and TLR $\sim$ SSS by Prop.~\ref{th:dlltrr}.

Consider now a shortest SSS invocation (equal in size to a shortest DLL invocation by Prop.~\ref{th:dlltrr}), referred to below as the ``original invocation''. If a 1UIP variable is never encountered, the claim is proven. Let $v_d$ be the first variable, discovered to be a 1UIP variable. Let SSS be at the point of making the decision $v_d \ass \sigma_d$, such that after flipping $v_d$ and backtracking, $v_d$ becomes a 1UIP variable. We show that there exists an ``updated'' SSS invocation of the same size, where the 1UIP-based CDB is not made at this point. Iterative applications of this principle result in receiving an SSS invocation, such that the condition for making a 1UIP-based CDB never holds; hence it can serve as an SSS with 1UIP-based CDB invocation of the same size as the original SSS invocation.

%If the condition for making a 1UIP-based CDB never holds during the original invocation, we are done, since a SSS with 1UIP-based CDB invocation making the same decisions would be valid. Let $v_d$ be the first variable, discovered to be a 1UIP variable in clause $C_{\textit{NewParent}}$ by the original invocation during backtracking. 

Let $g < d$ be the highest left decision level. The backtracking clause $C_{\textit{NewParent}}$ does not contain any literal assigned after $g$ by definition of the 1UIP-based CDB. Thus, $C_{Parent(d)}$ must contain at least one variable assigned after level $g-1$, since otherwise $g$ would be an obsolete decision; thus there would exist an invocation with fewer decisions than any shortest one. 

We claim that it is sufficient to make the decision $v_d \ass \lnot \sigma_d$ first. Explore the relevant (previously right) subtree of the original invocation; then flip the value of $v_d$ and explore the relevant (previously left) subtree of the original invocation to eliminate the situation when $v_d$ becomes a 1UIP variable. We ensure that new 1UIP variables are not created, the validity of the invocation is preserved and the number of decisions remains the same.

Suppose that the updated invocation has decided $v_d \ass \lnot \sigma_d$. A 1UIP-based CDB cannot happen before flipping $v_d$ and backtracking to $d$ in the updated invocation, since $v_d$ is the first 1UIP variable, discovered by the original invocation by construction. The clause $C_{\textit{NewParent}}$ will be generated as the parent clause for $d$. After the flip, the left subtree of the original invocation under $v_d \ass \sigma_d$ is explored. The newly generated resolvent clause $C_{\textit{NewParent}'}$ is identical to $C_{Parent(d)}$ of the original invocation; thus it must contain at least one variable, assigned after $g-1$. Hence, no 1UIP-based CDB is made at this point at present. The generated resolvent clause after backtracking beyond $d$ is exactly the same in the original and updated invocations. Making the same decisions as in the original invocation from this point on will result in a valid SSS invocation with the same number of decisions.
\qed
\end{proof}

%We will analyze the impact of 1UIP-based CDB on search pruning in Section~\ref{sec:eucdl}.

%As NCB, 1UIP-based CDB is helpful for backward tree pruning, described in Section~\ref{sec:eucdl}: parent resolution refutations of right decision levels between $g$ and $d$ are not included in the final derivation of the empty clause, unless CCR is used.

\subsection{The Power of a SAT Solver without Conflict Clause Recording}

Finally, we have that DLL is at least as strong as Algorithm~\ref{alg:dll} with Boolean Constraint Propagation, Non-Chronological Backtracking and 1UIP-based Conflict-Directed Backjumping.

\begin{proposition}\label{th:tlc}
DLL $\lesssim$ SSS + BCP, 1UIP-based CDB, NCB.
\end{proposition}
\begin{proof}
Any invocation of SSS with BCP, 1UIP-based CDB and NCB produces a TLR refutation, whose size is at most the number of decisions. 
\qed
\end{proof}

\section{A Modern SAT Solver Algorithm}
\label{sec:mssa}

To complete the picture of transforming Total Assignment Enumeration into a modern SAT solver, we need to define Conflict Clause Recording in our terminology. \emph{Conflict Clause Recording (CCR)} is an enhancement of DLL, allowing the algorithm to use some or all of the resolvent clauses for conflict identification (and propagation, if BCP is used). These clauses are called \emph{conflict clauses}. A similar approach to CCR was used in~\cite{ClRes,NieuwenhuisetalJACM2006}; however the literature on practical SAT solver design~\cite{Chaff2001,DBLP:conf/sat/EenS03,siege} uses an implication graph-based approach~\cite{GRASP1}. Our framework detaches CCR from other algorithms related to conflict analysis, such as Conflict-Directed Backjumping. For example, one can implement an algorithm using a 1UIP-based CDB, but without recording conflict clauses at all. With CCR, Algorithm~\ref{alg:dll} will still terminate with a valid complete resolution refutation; however, the refutation is no longer guaranteed to be tree-like. To turn Algorithm~\ref{alg:dll} into an algorithm that implements Chaff's scheme for CCR, it is sufficient to allow it to use all parent clauses for conflict identification: 

\begin{algorithmic}
\STATE \textbf{Parent-based Conflict Clause Recording} (invoked just before line~\ref{line:dis0})
\IF{$LRStatus(d)=\LL$ and $v_d^{\lnot \sigma_d} \in C_{\textit{NewParent}}$} \label{line:gfccr1}
\STATE $\textit{Instance} \ass Instance \cup \left\{\textit{NewParent}\right\}$ \label{line:gfccr2}
\ENDIF
\end{algorithmic}

%Note that if both 1UIP-based CDB and Flip-based CCR are used, then Flip-based CCR is invoked after the possible substitution of the decision variable of the maximal left decision level by the current UIP (as one would expect).

%1UIP-based CDB, NCB and Flip-based CCR with BCP constitute Chaff's Conflict-Driven Learning engine exactly. Moreover, the most modern SAT solvers, such as Minisat~\ref{} or Picosat~\ref{} all use this scheme, optionally enhanced by minimization, described in Section~\ref{}.

1UIP-based CDB, NCB and Parent-based CCR with BCP constitute Chaff's conflict analysis engine exactly. The most modern SAT solvers, such as   Minisat~\cite{DBLP:conf/sat/EenS03}, use this scheme, optionally enhanced by minimization~\cite{DBLP:conf/sat/EenS03}. A summary of steps that should be carried out to implement a modern SAT solver is provided in the conclusion.

\section{Conclusions}
\label{sec:concl}

We have proposed a new framework for presenting and analyzing the functionality of a modern DLL-based SAT solver. We described the following enhancements to the DLL algorithm: (1) Parent clause maintenance; (2) Boolean Constraint Propagation; (3) Non-Chronological Backtracking; (4) 1UIP-based Conflict-Directed Backjumping; (5) Parent-based Conflict Clause Recording. The above-mentioned algorithms are not interrelated in our approach. We exploited the inherent interrelation between backtracking and resolution not using the notion of implication graph. We demonstrated that DLL with 1UIP-based CDB, DLL with NCB, plain DLL and TLR are equally strong, and provided a family of formulas, whose shortest refutation in DLL with BCP is linearly longer than in DLL. We have also shown that parent clause maintenance, NCB, 1UIP-based CDB and BCP do make DLL stronger. We related the concept of search pruning to the size of the resolution refutation, derived by the algorithm, and pointed to the contribution of various algorithm to search pruning. 

The following is a suggestion how to implement a modern SAT solver: (1) Implement Algorithm~\ref{alg:dll} with Non-Chronological Backtracking (Sect.~\ref{sec:NCB}), 1UIP-based Conflict-Directed Backjumping (Sect.~\ref{sec:CDB}), Parent-based Conflict Clause Recording (Sect.~\ref{sec:mssa}) and Boolean Constraint Propagation (Sect.~\ref{sec:BCP}), using modern data structures~\cite{pMinisat}; (2) use a modern restart strategy, such as~\cite{DBLP:conf/sat/RyvchinS08}, and decision heuristic~\cite{DBLP:conf/sat/DershowitzHN05}. In the formal verification domain, use local conflict clause recording~\cite{DBLP:conf/sat/DershowitzHN07} and the implementation in \cite{Chaff2004} of decision stack shrinking~\cite{NadelThesis}; (3) use an efficient preprocessor before embarking on the search~\cite{DBLP:conf/sat/EenB05}.

The present work can serve as a basis for a future work on both the inference power of a modern SAT solver and on the practical SAT solver design. 
%\vspace*{-8pt}
\nopagebreak[0]
\bibliographystyle{plain}
\bibliography{2009}
\end{document}